\documentclass[submission,copyright,creativecommons]{eptcs}
\providecommand{\event}{NCL 2022} 
\usepackage{breakurl}             
\usepackage{underscore}           

\usepackage{amsthm}
\usepackage{amssymb}

\usepackage{tikz,subfigure}

\newcommand{\IMPL}{\rightarrow}

\newcommand{\AND}{\wedge}

\newcommand{\OR}{\vee}

\newcommand{\NOT}{\neg}

\newtheorem{theorem}{Theorem}
\newtheorem{lemma}{Lemma}
\newtheorem{definition}{Definition}

\newtheorem{proposition}{Proposition}

\title{Routley Star in Information-Based Semantics\thanks{This paper is an outcome of the project Logical Structure of Information Channels, no. 21-23610M, supported by the Czech Science Foundation and realized at the Institute of Philosophy of the Czech Academy of Sciences.}}
\author{V\'it Pun\v coch\'a\v r
\institute{Institute of Philosophy, \\Czech Academy of Sciences, \\The Czech Republic}
\email{puncochar@flu.cas.cz}
\and
Igor Sedl\' ar
\institute{Institute of Philosophy, \\Czech Academy of Sciences, \\The Czech Republic}
\email{sedlar@flu.cas.cz}
}
\def\titlerunning{Routley Star  in Information-Based Semantics}
\def\authorrunning{V. Pun\v coch\'a\v r, I. Sedl\' ar}
\begin{document}
\maketitle

\begin{abstract}
It is common in various non-classical logics, especially in relevant logics, to characterize negation semantically via the operation known as Routley star. This operation works well within relational semantic frameworks based on prime theories. We study this operation in the context of ``information-based'' semantics for which it is characteristic that sets of formulas supported by individual information states are theories that do not have to be prime. We will show that, somewhat surprisingly, the incorporation of Routley star into the information-based semantics does not lead to a collapse or a trivialization of the whole semantic system. On the contrary, it leads to a technically elegant though quite restricted semantic framework that determines a particular logic. We study some basic properties of this semantics. For example, we show that within this framework double negation law is valid only in involutive linear frames. We characterize axiomatically the logic of all linear frames and show that the logic of involutive linear frames coincides with a system that Mike Dunn coined Kalman logic. This logic is the $\{\AND, \OR, \NOT\}$-fragment of the ``semi-relevant'' logic known as $\mathsf{R}$-mingle. Finally, we characterize by a deductive system the logic of all information frames equipped with Routley star.
\end{abstract}

\section{Introduction}

This paper is a contribution to the general study of the operation known as Routley star, standardly denoted by the symbol $*$ (see \cite{Dunn93, Dunn99, Leitgeb19, Restall99}). Routley star is used in non-classical, especially relevant logics for the following semantic characterization of negation: $w \vDash \NOT \alpha$\ iff $w^* \nvDash \alpha$. This operation works so well because under some basic assumptions there is a corresponding operation on prime theories, mapping a prime theory $\Omega$ to the prime theory $\Omega^*=\{\alpha \mid \NOT \alpha \notin \Omega \}$ so that one obtains the analogue of the semantic clause  $\NOT \alpha  \in \Omega$ iff $\alpha \notin \Omega^*$. This fact is used in the canonical model constructions and completeness proofs for logics equipped with semantics using Routley star. Since primeness of the theories used in canonical models is crucial for this strategy it seems that Routley star can work well only within relational  semantic frameworks in which sets of formulas that are true at individual points of evaluation are always prime theories. In this paper we study Routley star in the context of a semantics that is based on theories rather than prime theories. We show that the incorporation of Routley star in this sematnics significantly restricts its flexibility but it does not lead to its collapse or a trivialization. On the contrary, it leads to a non-trivial semantics that determines a particular interesting logic.

We call our semantic framework \textit{information-based sematnics}. It has been used in the context of inquisitive semantics (see \cite{CiardelliRoelofsen11, Puncochar19, PuncocharSedlar21}).  Similar semantic frameworks were employed for example in \cite{Dosen89, Humberstone88, Wansing93}. The main idea behind this semantics is that there is a crucial conceptual difference between the notion of truth and the notion of informational support. Both can be viewed as relations between some specific entities on one side and formulas on the other side. In the case of truth the specific relatum is some ontic entity, a possible world or a situation. We will use the expression $w \vDash \alpha$ for the claim that the sentence $\alpha$ is true in a world (or situation) $w$. In the case of informational support the specific relatum is a body of information or an information state  and we will write $s \Vdash \alpha$ expressing that $\alpha$ is supported by the information state $s$. Truth and informational support seem to be often confused but arguably they are different notions.

The clearest way to show this is to compare the natural truth and support conditions for disjunction. It makes a perfect sense to claim that a disjunction $\alpha \OR \beta$ is true in a world (or situation) $w$ if and only if $\alpha$ or $\beta$\ is true in $w$. However, informational support of disjunction cannot be characterized by an analogous condition. It is quite common that a body of information supports a disjunction without supporting any of its disjuncts. Technically speaking, this means that the sets of formulas supported by individual information states are theories that do not have to be prime. Our aim here is to study Routley star negation in the context of such information-based semantics (a different treatment of negation in information-based semantics, characterized in terms of an incompatibility relation, was studied in \cite{PuncocharTedder21}). Our main result is a syntactic characterization of the logic of all information frames equipped with Routley star.

The article is structured as follows. In Section \ref{sec:rsstandard} we recall the Routley star in the context of the standard framework where points of evaluation correspond to prime theories. Section \ref{sec:infosem} introduces information-based semantics with the Routley star and proves a series of characterization results; in particular, it is shown that the double negation law forces linearity of frames. Section \ref{sec:linframes} studies linear frames. In particular, the axiomatizations of the logic of all linear frames and all involutive linear frames, respectively, are provided. The latter logic coincides with what Dunn called Kalman logic. Section \ref{sec:logicofallRIF} establishes a completeness result for the logic of all information frames with the Routley star.

\section{Routley star in the standard framework}\label{sec:rsstandard}

Let $L$ be the language built up from a set of atomic formulas $At$ using conjunction $\AND$, disjunction $\OR$ and negation $\NOT$. We will now reconstruct the standard semantic framework used for this language in non-classical logics. For more details, see \cite{Dunn93, Dunn99}.

 A standard frame (SF, for short) is a tuple $\mathcal{F}=\langle W, \leq, * \rangle$ where $\langle W, \leq \rangle$ is a partially ordered set and $*$ is a unary operation on $W$ such that $v \leq w$ implies $w^* \leq v^*$, for all $v, w \in W$.  A standard model (SM) is a pair $\mathcal{M}=\langle \mathcal{F}, V\rangle$, where $\mathcal{F}$ is an SF and $V$ is a valuation defined as a function assigning to each atomic formula an upward closed set in $\mathcal{F}$. Given such a model $\mathcal{M}$ the following semantic clauses specify recursively a relation of truth $\vDash$  between the elements of $\mathcal{F}$ and formulas of $L$:
\begin{itemize}
\item[] $w \vDash p$ iff $w \in V(p)$, for every $p \in At$,
\item[] $w \vDash \alpha \AND \beta$ iff $w \vDash \alpha$ and $w \vDash \beta$,
\item[] $w \vDash \alpha \OR \beta$ iff $w \vDash \alpha$ or $w \vDash \beta$,
\item[] $w \vDash \NOT \alpha$ iff $w^* \nvDash \alpha$.
\end{itemize}
A crucial feature of the clauses is that they yield persistence, the property required for atomic formulas in the definition of valuation: for any $L$-formula $\alpha$, if $w \vDash \alpha$\ and $w \leq v$ then $v \vDash \alpha$. In other words, the propositions expressed by formulas are always upward closed sets and the algebra of propositions is the algebra of such sets.

A pair of $L$-formulas $\left\langle \alpha, \beta \right\rangle$ will be called a (consequence) $L$-pair. We will write $\alpha \vdash \beta$ instead of $\left\langle \alpha, \beta \right\rangle$. We say that $\alpha \vdash \beta$ is valid in $\mathcal{M}$ if for every $w$\ in $\mathcal{M}$ if $w \vDash \alpha$ then $w \vDash \beta$. We say that $\alpha \vdash \beta$ is valid in a class of standard models if it is valid in every model of that class. 

It is known that the set of $L$-pairs valid in all standard models is generated by the deductive system consisting of the axioms and rules of the distributive lattice logic $\mathsf{DLL}$:

\begin{center}
\begin{tabular}{ccccc}
(ID) $\alpha \vdash \alpha$ & (e$\AND_1$) $\alpha \AND \beta \vdash \alpha$ & (e$\AND_2$) $\alpha \AND \beta \vdash \beta$ & (i$\OR_1$) $\alpha \vdash \alpha \OR \beta$ & (i$\OR_2$) $\beta \vdash \alpha \OR \beta$ \end{tabular}
\end{center}

\begin{center}
\begin{tabular}{c}
(D) $\alpha \AND (\beta \OR \gamma) \vdash (\alpha \AND \beta) \OR (\alpha \AND \gamma)$
\end{tabular}
\end{center}

\begin{center}
\begin{tabular}{ccc}
(T) $\alpha \vdash \beta, \beta \vdash \gamma / \alpha \vdash \gamma$ & (i$\AND$) $\alpha \vdash \beta, \alpha \vdash \gamma / \alpha \vdash \beta \AND \gamma$ & (e$\OR$) $\alpha \vdash \gamma, \beta \vdash \gamma / \alpha \OR \beta \vdash \gamma$ \\
\end{tabular}
\end{center}

plus the following axioms and rule characterizing the behaviour of negation:

\begin{center}
\begin{tabular}{cc}
(DM1) $\NOT (\alpha \AND \beta) \vdash \NOT \alpha \OR \NOT \beta$ & (DM2) $\NOT \alpha \AND \NOT \beta \vdash \NOT (\alpha \OR \beta)$  \\
\end{tabular}
\end{center}

\begin{center}
\begin{tabular}{c}
(N) $\alpha \vdash \beta / \NOT \beta \vdash \NOT \alpha$ \\
\end{tabular}
\end{center}

We can denote this system, or more precisely, the set of $L$-pairs generated by this system, as \textit{Distributive Lattice Logic with Routley Star Negation}, or $\textsf{DLLR}$, for short.

The elegance of the treatment of negation via the Routley star is clear form the completeness proof which is based on the standard canonical model construction. The canonical model is build out of prime theories. A set of $L$-formulas $\Omega$ is a prime $\textsf{DLLR}$-theory if it satisfies the following three conditions: (a)~if $\alpha \in \Omega$ and $\alpha \vdash \beta \in \textsf{DLLR}$ then $\beta \in \Omega$; (b)  if $\alpha \in \Omega$ and $\beta \in \Omega$ then $\alpha \AND \beta \in \Omega$; (c) if $\alpha \OR \beta \in \Omega$ then $\alpha \in \Omega$ or $\beta \in \Omega$. Note that the set of all formulas that are true at a given $w$ is always a prime $\textsf{DLLR}$-theory.

Now it is crucial that there is an operation on prime $\textsf{DLLR}$-theories that corresponds to the Routley star operation. In particular, for any prime $\textsf{DLLR}$-theory $\Omega$ let us define the set of formulas $\Omega^*=\{ \beta~|~\NOT \beta \notin \Omega \}$. One can easily observe that $\Omega^*$ is again a prime $\textsf{DLLR}$-theory (to show this the axioms (DM1) and (DM2) and the rule (N) are needed). Moreover, it follows directly from the definition of $*$ that if $\Omega$, $\Gamma$ are prime $\textsf{DLLR}$-theories such that $\Omega \subseteq \Gamma$, then $\Gamma^* \subseteq \Omega^*$. So, one can take the set of prime $\textsf{DLLR}$-theories ordered by inclusion and equipped with $*$ and the resulting structure is indeed a standard frame. Moreover, if valuation is defined in the expected way ($\Omega \in V(p)$ iff $p \in \Omega$) one can prove by induction the truth lemma: $\Omega \vDash \alpha$\ iff $\alpha \in \Omega$, from which completeness follows by a Lindenbaum-style extension lemma.

If the standard frames are required to satisfy generally the property $w=w^{**}$ then the resulting logic is axiomatized by the system above enriched with the double negation laws (DN1) $\alpha \vdash \NOT \NOT \alpha$ and (DN2)  $\NOT \NOT \alpha \vdash \alpha$. Classical logic is then obtained semantically by requiring  $w=w^{*}$ (from which it follows that $\leq$\ is reduced to the identity relation) and syntactically by enriching the logic with the law of excluded middle (EM) $\beta \vdash \alpha \OR \NOT \alpha$ and the principle of explosion (EFQ) $\alpha \AND \NOT \alpha \vdash \beta$.



\section{Information-based semantics}\label{sec:infosem}

The standard semantics presented in the previous section was based on truth conditions as is clear from the semantic clause for disjunction. We now want to move to a semantics in which the central notion is informational support rather than truth. So, we must lift the assumption that a disjunction holds at a point of evaluation only if some of its disjuncts hold at this point. In other words, we must lift the assumption that sets of formulas that hold at particular points of evaluation are always prime theories. This undermines the use of Routley star. The goal of this paper is to explore whether an information-based semantics with the Routley star treatment of negation completely breaks down or whether it leads to a non-trivial framework. Somewhat surprisingly, it will turn out that the resulting framework leads to an interesting though rather restricted semantics. Let us introduce the frames of information-based semantics.

\begin{definition}
A Routley information proto-frame (or just Routley proto-IF) is a structure $\mathcal{I}=\linebreak\left\langle S, \circ, *,  i, e\right\rangle$, where $\left\langle S, \circ \right\rangle$ is a semilattice, i.e. $S$ is a non-empty set and $\circ$\ is an associative, commutative and idempotent binary operation on $S$; moreover, $*: S \IMPL S$ and $i, e$ are two distinct elements of $S$ for which the following conditions are required: (a) $s \circ i = s$, for any $s \in S$, (b) $s \circ e = e$, for any $s \in S$, and (c) for any $t, u \in S$, if $e = t \circ u$ then $e=t$ or $e=u$.
\end{definition}

We define an order on $\mathcal{I}$ in the following way: $s \leq t$ iff $s \circ t = s$. It follows that $e$ is the least element, and $i$ is the top element of the structure.

The set $S$\ represents a space of information states. The operation $\circ$ assigns to any two states $t, u$ the common  (informational) content of $t$\ and $u$, i.e. a state $t \circ u$ consisting of information that the two states have in common. The element $i$\ represents the state of absolute inconsistency, and the state $e$ the state of absolute ignorance. That is, $i$ is intended to satisfy every formula, and $e$ no formula. The informal claim that $t \circ u$\ is the common informational content of $t$ and $u$ translates into the formal requirement that for any formula $\alpha$, $t \circ u$ satisfies $\alpha$\ iff both $t$\ and $u$ satisfy $\alpha$. In other words, the set of states satisfying a formula should form a (proper) filter in $\mathcal{I}$. This determines a restriction on valuations.

 \begin{definition}
Let $\mathcal{I}=\left\langle S, \circ, *,  i, e\right\rangle$ be a Routley proto-IF. A proper filter in $\mathcal{I}$ is any subset $F$ of $S$\ such that $i \in F$, $e \notin F$, and for any $t, u \in S$ it holds that $t \circ u \in F$ iff $t \in F$ and $u \in F$.  A valuation in $\mathcal{I}$ is a function that assigns to each atomic formula from $At$ a proper filter in $\mathcal{I}$. A Routley information proto-model (or just Routley proto-IM) is a Routley proto-IF equipped with a valuation.
\end{definition}

Note that one can equivalently define filters as non-empty proper subsets of $S$ that are upward closed (w.r.t. $\leq$) and closed under $\circ$. Relative to a given Routley proto-IM, we can introduce the following semantic clauses for $L$-formulas (the non-standard semantic clause for disjunction captures formally the idea that disjunction is what two states supporting the disjuncts have in common):

\begin{itemize}
\item[] $s \Vdash p$\ iff $s \in V(p)$, for any atomic formula $p$,
\item[] $s \Vdash \alpha \AND \beta$\ iff $s \Vdash \alpha$ and $s \Vdash \beta$,
\item[] $s \Vdash \alpha \OR \beta$\ iff $t \circ u \leq s$ for some states $t, u$ such that $t \Vdash \alpha$ and $u \Vdash \beta$,
\item[] $s \Vdash \NOT \alpha$\ iff $s^* \nVdash \alpha$.
\end{itemize}
If $s \Vdash \alpha$, we say that the state $s$\ supports the formula $\alpha$. Let $\mathcal{N}$ be a proto-IM, and $\alpha$ an $L$-formula.  Then $|| \alpha ||_{\mathcal{N}}$ is the set of states in $\mathcal{N}$\ that support $\alpha$ (the subscript will usually be omitted). We will call it the proposition expressed by $\alpha$ in $\mathcal{N}$. 

We want to have  that propositions expressed by $L$-formulas are always proper filters, which corresponds to the claim that the set of formulas supported by a state $s \circ t$, i.e. by the common content of $s$\ and $t$, is the intersection of the set of formulas supported by $s$ and the set of formulas supported by $t$. For atomic formulas this property is guaranteed directly by the definition of a valuation. The semantic clauses for conjunction and disjunction form proper filters from proper filters. However, the semantic clause for negation does not generally preserve this property. We will restrict ourself to the largest class of Routley proto-IFs where this property is preserved. This class is described in the following proposition.

\begin{proposition}\label{p:charoffilters}
Let $\mathcal{I}$ be a Routley proto-IF. Then $|| \alpha ||$ is a proper filter in $\left\langle \mathcal{I}, V \right\rangle$, for every valuation $V$ and every  $L$-formula $\alpha$, iff the following conditions are satisfied:
\begin{itemize}
\item[(a)] $i^* = e$ and $e^* = i$,
\item[(b)] for any $t, u \in S$, if $t \leq u$ then $u^* \leq t^*$,
\item[(c)] for any $t, u \in S$, $(t \circ u)^*\leq t^*$ or $(t \circ u)^* \leq u^*$. 
\end{itemize}
\end{proposition}
\begin{proof}
Assume that the conditions (a)-(c) are satisfied. We need to show that for any $L$-formula $\alpha$, if $|| \alpha ||$ is a proper filter, then $|| \NOT \alpha ||$ is a proper filter. Assume that $|| \alpha ||$ is a proper filter. Since $i \vDash \alpha$ and $e^* = i$, we obtain $e \notin || \NOT \alpha ||$. Since $e \nvDash \alpha$\ and $i^* = e$, we obtain $i \in || \NOT \alpha ||$.  Now, we show that  $|| \NOT \alpha ||$ is closed under $\leq$ and $\circ$. First, assume that $t \vDash \NOT \alpha$ and $t \leq u$. Then $u^* \leq t^*$, and $t^* \nvDash \alpha$. So, $u^* \nvDash \alpha$, i.e. $u \vDash \NOT \alpha$. Second, assume that $t  \vDash \NOT \alpha$ and $u  \vDash \NOT \alpha$, i.e. $t^* \nvDash \alpha$ and $u^* \nvDash \alpha$. Then, due to (c), we obtain $(t \circ u)^* \nvDash \alpha$, i.e. $t \circ u \vDash \NOT \alpha$.

Now, we will assume that some of the conditions (a)-(c) is not satisfied and we will show that we can define a valuation such that the formula $\NOT p$ will not semantically correspond to a filter. 

Assume that $i^* \neq e$ and take a valuation $V$\ such that $V(p)=\{ s~|~i^* \leq s \}$. The assumption $i^* \neq e$ implies that $\{ s~|~i^* \leq s \}$ is a proper filter and thus such a valuation exists. Then $i^* \vDash p$\ and so\ $i \nVdash \NOT p$. Now assume that  $e^* \neq i$  and take a valuation $V$\ such that $V(p)=\{ i \}$. Then $e^* \nvDash p$\ and thus $e \vDash \NOT p$. 

Next assume that (b) is not satisfied. So, there are $t, u \in S$ such that $t \leq u$ but not $u^* \leq t^*$. Take a valuation $V$\ such that $V(p)=\{ s~|~u^* \leq s \}$. Then $u^* \vDash p$, but $t^* \nvDash p$, i.e. $u \nvDash \NOT p$, but $t \vDash  \NOT p$. 

Finally, assume that (c) is not satisfied. We can also assume that (b) is satisfied (if not, we can proceed as in the previous paragraph). Since (c) is not satisfied, there are $t, u \in S$ such that neither $(t \circ u)^* \leq t^*$ nor $(t \circ u)^* \leq u^*$. Take a valuation $V$\ such that $V(p)=\{ s~|~(t \circ u)^* \leq s \}$. Then $(t \circ u)^* \vDash p$ but $t^* \nvDash p$\ and $u^*\nvDash p$. So, $t \vDash \NOT p$, $u \vDash \NOT p$ but $t \circ u \nvDash \NOT p$.
\end{proof}
In the light of the previous proposition it is desirable to introduce the following definition. 
\begin{definition}
Any Routley proto-IF (proto-IM) that satisfies the conditions (a)-(c) from Proposition \ref{p:charoffilters} will be called a Routley IF (IM). 
\end{definition}

Note that due to Proposition \ref{p:charoffilters}, we obtain that the state $i$ in any Routley IM supports every formula, and the state $e$ supports no formula. This also explains the role of these two elements in the semantics. The presence of $i$ guarantees that if a state $s$ supports $\alpha$ then $s$\ supports also $\alpha \OR \beta$ because $i$\ supports $\beta$ and $s = s \circ i$. The empty state $e$ is needed as the $*$-image of the state\ $i$.

In general disjunction can be supported by a state even if none of the disjuncts is. But the properties of Routley IFs lead to an asymmetry between positive formulas and their negations.  It follows from the definition of Routley IFs that disjunctions of negations can be characterized by the truth-conditional semantic clause.

\begin{proposition}\label{p:disjunctionofnegations}
For any Routley IM, any state $s$\ in that model, and all $L$-formulas $\alpha, \beta$:
\begin{itemize}
\item[] $s \Vdash \NOT \alpha \OR \NOT \beta$ iff $s \Vdash \NOT \alpha$ or $s \Vdash \NOT \beta$.
\end{itemize}
\end{proposition}
\begin{proof}
We need to show only the left-to-right implication. Assume $s \Vdash \NOT \alpha \OR \NOT \beta$. Then there are states $t, u$ such that $t^* \nVdash \alpha$, $u^* \nVdash \beta$ and $t \circ u \leq s$. It follows that $s^* \leq (t \circ u)^*$, and thus $s^* \leq t^*$ or  $s^* \leq u^*$. Hence, $s^*  \nVdash \alpha$ or $s^*  \nVdash \beta$, that is  $s \Vdash \NOT \alpha$ or $s \Vdash \NOT \beta$.
\end{proof}

The condition (c) from the definition of Routley IFs is quite strong and restricting, as shown by the following proposition.

\begin{proposition}\label{p:linearityofstar}
Let $\mathcal{I}$ be a Routley IF. Then it holds:
\begin{itemize}
\item[(a)]  for every $t, u \in S$, $(t \circ u)^*= t^*$ or $(t \circ u)^* = u^*$,
\item[(b)]  for every $t, u \in S$, $t^* \leq u^*$ or $u^* \leq t^*$,
\item[(c)]  for every $t \in S$, $t^* \leq t^{**}$ or $t^{**} \leq t^{*}$.
\end{itemize}
\end{proposition}
\begin{proof}
(a) Since $t \circ u \leq t$\ and $t \circ u \leq u$ we have $t^* \leq (t \circ u)^*$ and $u^* \leq (t \circ u)^*$. Moreover, $(t \circ u)^*\leq t^*$ or $(t \circ u)^* \leq u^*$, and thus $(t \circ u)^*= t^*$ or $(t \circ u)^* = u^*$.

(b) Since $t^* \leq (t \circ u)^*$ and $u^* \leq (t \circ u)^*$, and $(t \circ u)^*= t^*$ or $(t \circ u)^* = u^*$, we have $t^* \leq u^*$ or $u^* \leq t^*$. 

(c) This claim is obtained just by applying (b) to $t$\ and $t^*$.
\end{proof}

We say that an $L$-pair $\alpha \vdash \beta$ is valid in a given Routley IM $\mathcal{N}$ if $|| \alpha || \subseteq || \beta ||$. An $L$-pair is valid in a Routely IF $\mathcal{I}$ if it is valid in every Routely IM based on $\mathcal{I}$. An $L$-pair is valid in a class of Routely IFs (IMs) if it is valid in every member of that class.

\begin{proposition}\label{p:somevaliditiesinifs}
Let $\mathcal{N}$ be a Routley IM. Then it holds for all $L$-formulas $\alpha$, $\beta$:
\begin{itemize}
\item[(a)] The schematic rules (T), (i$\AND$), (e$\OR$), (N) preserve validity in $\mathcal{N}$, i.e. if the premises are valid in $\mathcal{N}$ then the conclusion is valid in $\mathcal{N}$ as well.
\item[(b)] All instances of the axiom schemata (ID), (e$\AND_1$), (e$\AND_2$), (i$\OR_1$), (i$\OR_2$), (DM1) are valid in $\mathcal{N}$.
\item[(c)]  All instances of the axiom schemata (DM2$^*$) $\NOT \NOT \alpha \AND \NOT \NOT \beta \vdash \NOT (\NOT \alpha \OR \NOT \beta)$ and (D$^*$) $\alpha \AND (\NOT \beta \OR \NOT \gamma) \vdash (\alpha \AND \NOT \beta) \OR (\alpha \AND \NOT \gamma)$ are valid in $\mathcal{N}$.
\end{itemize}
\end{proposition}
\begin{proof}
(a) We will show that (N) preserves validity: Assume that $\alpha \vdash \beta$\ is valid in $\mathcal{N}$. Assume for some state $s$, that $s \Vdash \NOT \beta$. Then $s^* \nVdash \beta$, and thus $s^* \nVdash \alpha$. Hence, $s \Vdash \NOT \alpha$.

(b) We will show the validity of (DM1): Assume for some state $s$, that $s \Vdash \NOT(\alpha \AND \beta)$, i.e. $s^* \nVdash \alpha \AND \beta$. Then $s^* \nVdash \alpha$\ or $s^* \nVdash \beta$, i.e. $s \Vdash \NOT \alpha$ or $s \Vdash \NOT \beta$. Hence $s \Vdash \NOT \alpha \OR \NOT \beta$.

(c) First, assume for some state $s$, that $s \Vdash \NOT \NOT \alpha \AND \NOT \NOT \beta$, i.e. $s^{*} \nVdash \NOT \alpha$ and $s^{*} \nVdash \NOT \beta$. It follows from Proposition \ref{p:disjunctionofnegations} that $s^{*} \nVdash \NOT \alpha \OR \NOT \beta$, i.e. $s \Vdash \NOT (\NOT \alpha \OR \NOT \beta)$. Second, assume $s \Vdash \alpha \AND (\NOT \beta \OR \NOT \gamma)$, i.e. $s \Vdash \alpha$ and $s \Vdash \NOT \beta \OR \NOT \gamma$. So, $s \Vdash \NOT \beta$ or $s \Vdash \NOT \gamma$, by  Proposition \ref{p:disjunctionofnegations}. If $s \Vdash \NOT \beta$, then $s \Vdash \alpha \AND \NOT \beta$, and hence $s \Vdash (\alpha \AND \NOT \beta) \OR (\alpha \AND \NOT \gamma)$. If $s \Vdash \NOT \gamma$, we can reason in an analogous way.
\end{proof}

The claim (b) in Proposition \ref{p:linearityofstar} shows that in any  Routley IF the image of the operation $*$ forms a linearly ordered set. However, the whole frame does not have to be linearly ordered. Consider the two non-linear examples of  RIFs from Fig. \ref{f: counterexamples1}.
\begin{figure}
\begin{center}
\begin{tabular}{cc}
\begin{tikzpicture}
\node (a) at (2,0) {$s$};
\node (b) at (1,1) {$t$};
\node (c) at (3,1) {$u$};
\node (d) at (2,2) {$v$};
\node (e) at (2,-0.6) {$e$};
\node (f) at (2,3.2) {$i$};

\draw[->] (a) -- (d);
\draw[->] (b) -- (d);
\draw[->] (c) -- (d);
\draw[->] (d) edge [loop above] (d);
\draw[->] (e) edge [bend left=80] (f);
\draw[->] (f) edge [bend left=80] (e);
\end{tikzpicture}

&

\begin{tikzpicture}
\node (s) at (2,0) {$s$};
\node (t) at (1,1) {$t$};
\node (u) at (3,1) {$u$};
\node (v) at (2,2) {$v$};
\node (e) at (2,-1.2) {$e$};
\node (i) at (2,2.6) {$i$};

\draw[->] (v) -- (s);
\draw[->] (t) -- (s);
\draw[->] (u) -- (s);
\draw[->] (s) edge [loop below] (d);
\draw[->] (e) edge [bend left=80] (i);
\draw[->] (i) edge [bend left=80] (e);
\end{tikzpicture}
\end{tabular}
\end{center}
\caption{Examples of non-linear Roultey IFs}\label{f: counterexamples1}
\end{figure}
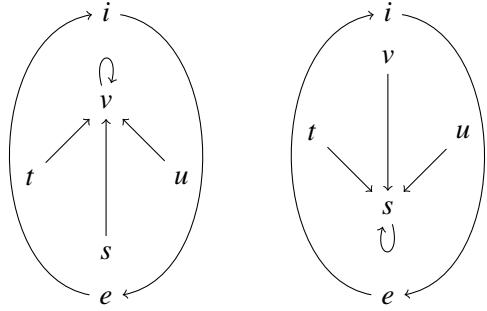
The vertical position determines the order so that, for example, in both cases $s \leq t$ and $u \leq v$. The states $t, u$ are the only incomparable states. It holds that $t \circ u =s$. The arrows determine the behaviour of $*$ so that, for example, on the left $t^*=v$, and on the right $t^*=s$.  These frames can also be used to construct counterexamples to various consequence $L$-pairs. For example, we can observe the following:
\begin{itemize}
\item $\NOT \NOT p \vdash p$\ is not valid in the frame on the left (assume that $V(p)= \{v, i \}$; then $t \vDash \NOT \NOT p$\ but $t \nvDash p$),
\item $p \vdash \NOT \NOT p$\ is not valid in the frame on the right (assume that $V(p)= \{v, i \}$; then $v \vDash p$\ but $v \nvDash \NOT \NOT p$),
\item $\NOT p \AND \NOT q \vdash \NOT (p \OR q)$\ is not valid in the frame on the right (assume that $V(p)= \{t, v, i \}$ and $V(q)= \{u, v, i \}$; then $s \vDash \NOT p \AND \NOT q$\ but $s \nvDash \NOT (p \OR q)$).
\end{itemize}
From the axioms of the distributive lattice logic $\mathsf{DLL}$ we lose the distributive axiom (D), which is not generally valid in the class of all Roultey IFs. This can be illustrated with the example in Fig. \ref{f: counterexamples2}. The states $s, t, u, v, w$\ are ordered as in the non-distributive lattice $N_5$, and the Routley star, for instance, maps $i$ to $e$, $e$ to $i$, and every other element to $u$ (Routley star does not play any role in the example, it is introduced just to make the structure a Routley IF).
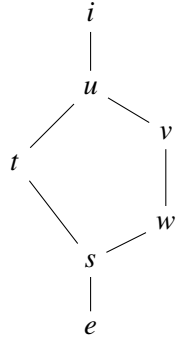
\begin{figure}
\begin{center}
\begin{tabular}{c}
\begin{tikzpicture}
\node (s) at (2,-0.3) {$s$};
\node (t) at (1,1) {$t$};
\node (v) at (3,1.4) {$v$};
\node (w) at (3, 0.2) {$w$};
\node (u) at (2,2) {$u$};
\node (e) at (2,-1.2) {$e$};
\node (i) at (2,3) {$i$};

\draw[-] (e) -- (s);
\draw[-] (s) -- (t);
\draw[-] (s) -- (w);
\draw[-] (w) -- (v);
\draw[-] (v) -- (u);
\draw[-] (t) -- (u);
\draw[-] (u) -- (i);
\end{tikzpicture}
\end{tabular}
\end{center}
\caption{A counterexample to distributivity}\label{f: counterexamples2}
\end{figure}
For a counterexample to $p \AND (q \OR r) \vdash (p \AND q) \OR (p \AND r)$ it suffices to take the valuation $V$ such that $V(p)$ is the principle filter generated by $w$,  $V(q)$ is the principle filter generated by $t$,\ and $V(r)$ is the principal filter generated by $v$. In the resulting model $w$  supports $p \AND (q \OR r)$ but not $(p \AND q) \OR (p \AND r)$.

\begin{definition} 
We say that a consequence $L$-pair $\alpha \vdash \beta$ characterizes a class of Roultey IFs $\mathcal{C}$ if it holds for any Routley IF $\mathcal{I}$ that $\alpha \vdash \beta$ is valid in $\mathcal{I}$ iff $\mathcal{I} \in \mathcal{C}$. 
\end{definition}

\begin{proposition}\label{p:chardndemorg}
The following holds:
\begin{itemize} 
\item[(a)] $p \vdash \NOT \NOT p$\ characterizes the class of Routley IFs satisfying $s \leq s^{**}$, for each state $s$,
\item[(b)] $\NOT \NOT p \vdash p$\ characterizes the class of Routley IFs satisfying $s^{**} \leq s$, for each state $s$,
\item[(c)] $\NOT p \AND \NOT q \vdash \NOT (p \OR q)$\ characterizes the class of Routley IFs satisfying the following property: for all states $s, t, u$, if $t \circ u \leq s^*$ then $t \leq s^*$ or $u \leq s^*$.
\end{itemize}
\end{proposition}
\begin{proof}
(a) Assume that $s \leq s^{**}$, for each state $s$. Assume $s \vDash \alpha$. Then $s^{**} \vDash \alpha$, i.e. $s \vDash \NOT \NOT \alpha$. Now assume that there is a state $s$\ such that $s \nleq s^{**}$. Take a valuation $V$ such that $V(p)=\{t~|~s \leq t \}$. Then $s^{**} \nvDash p$, and so $s \nvDash \NOT \NOT p$. But $s \vDash p$.

(b) Assume that $s^{**} \leq s$, for each state $s$. Assume $s \vDash \NOT \NOT \alpha$. Then $s^{**} \vDash \alpha$, and hence $s \vDash \alpha$. Now assume that there is a state $s$ such that $s^{**} \nleq s$. Take a valuation $V$ such that $V(p)=\{t~|~s^{**} \leq t \}$. Then $s \nvDash p$. But $s \vDash \NOT \NOT p$.

(c) First, assume that for all states $s, t, u$, if $t \circ u \leq s^*$ then $t \leq s^*$ or $u \leq s^*$. Assume $s \nvDash \NOT (\alpha \OR \beta)$, i.e. $s^* \vDash \alpha \OR \beta$. So, there are $t, u$ such that $t \vDash \alpha$, $u \vDash \beta$, and $t \circ u \leq s^*$. Then $t \leq s^*$ or $u \leq s^*$ and thus $s^* \vDash \alpha$ or $s^* \vDash \beta$. Hence $s \nvDash \NOT \alpha \AND \NOT \beta$. Second, assume that there are states $s, t, u$\ such that $t \circ u \leq s^*$\ but neither $t \leq s^*$ nor $u \leq s^*$. Take a valuation $V$ such that $V(p)=\{v~|~t \leq v \}$ and $V(q)=\{v~|~u \leq v \}$. So, $ s^* \nvDash p$ and  $s^* \nvDash q$, i.e. $s \vDash \NOT p \AND \NOT q$. But since $t \circ u \vDash p \OR q$, we obtain $s^* \vDash p \OR q$, and so $s \nvDash \NOT (p \OR q)$. 
\end{proof}

Note that the frame pictured in Fig. \ref{f: counterexamples1} on the left satisfies the conditions (a) and (c) but not (b), from Proposition \ref{p:chardndemorg}, and the frame on the right satisfies (b) but neither  (a) nor (c). As a consequence of Proposition \ref{p:chardndemorg} we obtain the following rather surprising fact.

\begin{proposition}
Let $\mathcal{I}$ be a Routley IF.  If $p \vdash \NOT \NOT p$ and $\NOT \NOT p \vdash p$ are both valid in  $\mathcal{I}$ then $\mathcal{I}$ is linearly ordered by $\leq$.
\end{proposition}
\begin{proof}
If  $p \vdash \NOT \NOT p$ and $\NOT \NOT p \vdash p$ are both valid in  $\mathcal{I}$, then according to (a) and (b) of Proposition \ref{p:chardndemorg}, $s = s^{**}$, for every state $s$. So,  the image of $*$ is the whole set of states. Then, according to (b) of Proposition \ref{p:linearityofstar}, $\mathcal{I}$ must be linearly ordered.
\end{proof}
The converse of the previous proposition is not true since there are linearly ordered IFs in which $s = s^{**}$ does not hold generally. These frames are discussed in the following section. 

\section{Linearly ordered Routley IFs}\label{sec:linframes}

Linearly ordered Routley IFs (i.e. those Routley IFs in which $\leq$ is a linear order) are special in that they are in the common area of the information-based semantics and the standard semantics. The reason is that in the linearly ordered Routley IFs our semantic condition for disjunction reduces to the standard one.
\begin{proposition}
For any linear Routley IM, any state $s$\ in that model, and every $L$-formulas $\alpha, \beta$:
\begin{itemize}
\item[] $s \vDash \alpha \OR \beta$ iff $s \vDash \alpha$ or $s \vDash \beta$.
\end{itemize}
\end{proposition}
\begin{proof}
In any  linearly ordered Routley IF, $t \circ u= t$ or $t \circ u= u$.
\end{proof}

Let us characterize the logic of all linearly ordered Routley IFs by a deductive system. For any natural number $n$\ let $\NOT^n \alpha$\ denote the formula $\NOT \ldots \NOT \alpha$ with $n$ occurrences of $\NOT$ (in particular, $\NOT^0 \alpha=\alpha$). We will also need the following notation. If $n$ is a natural number and $s$ is a state in a Routley IM, then $s^{n(*)}$ is the state $s^{* \ldots *}$ with $n$ occurrences of $*$ (in particular, $s^{0(*)}=s$). Note that for any natural number $k$, $s^{2k(*)} \Vdash \alpha$ iff $s \Vdash \NOT^{2k(*)}\alpha$, and $s^{2k+1(*)} \Vdash \alpha$ iff $s \nVdash \NOT^{2k+1(*)}\alpha$.

Consider the deductive system consisting of the axioms and rules of $\mathsf{DLLR}$ enriched with the following two axioms for each natural number $k$:
\begin{center}
\begin{tabular}{cc}
(L1) $\alpha \AND \NOT^{2k} \beta \vdash  \NOT^{2k}  \alpha \OR \beta$ & (L2) $\alpha \AND \NOT^{2k+1} \alpha \vdash  \beta \OR \NOT^{2k+1} \beta$  \\
\end{tabular}
\end{center}
Let us denote this logic, or more precisely the set of $L$-pairs derivable in the resulting system, as $\mathsf{LinDLLR}$.

\vspace{0.2cm}

\begin{theorem}\label{p:compforlinrifs}
$\alpha \vdash \beta$ is valid in all linearly ordered Routley IFs iff $\alpha \vdash \beta \in \mathsf{LinDLLR}$.
\end{theorem}
\begin{proof}
The right-to-left implication amounts to soundness of the system with respect to linear Routley IFs. Let us check the soundness of the two axioms (L1) and (L2).

(L1) Let $s \Vdash \alpha \AND \NOT^{2k} \beta$\ in a linear Routley IM, that is $s \Vdash \alpha$\ and $s^{2k(*)} \Vdash \beta$. We will consider two possible cases. First, assume that $s \leq s^{2k(*)}$. Then $s^{2k(*)} \Vdash \alpha$, and hence $s \Vdash  \NOT^{2k}\alpha$. It follows that $s \Vdash  \NOT^{2k}\alpha \OR \beta$. Second, assume that $s^{2k(*)} \leq s$. Then $s \Vdash \beta$ and thus  $s \Vdash  \NOT^{2k}\alpha \OR \beta$.

(L2) Let $s \Vdash \alpha \AND \NOT^{2k+1} \alpha$ in a linear Routley IM, that is $s \Vdash \alpha$\ and $s^{2k+1(*)} \nVdash \alpha$. It follows that $s \nleq s^{2k+1(*)}$ and hence $s^{2k+1(*)} \leq s$. If $s \Vdash \beta$ then also $s \Vdash  \beta \OR \NOT^{2k+1} \beta$. If $s \nVdash \beta$\ then $s^{2k+1(*)}  \nVdash \beta$\ and so $s \Vdash \NOT^{2k+1} \beta$. Hence, we obtain  also in this case $s \Vdash  \beta \OR \NOT^{2k+1} \beta$.
  
Now we will prove the left-to-right implication, i.e. completeness of the system. For any set of formulas  $\Omega$\ we define $\Omega^*=\{\alpha \in L \mid \NOT \alpha \notin \Omega \}$. Note that if $\Omega$ is a prime $\mathsf{LinDLLR}$-theory then $\Omega^*$ is also a prime $\mathsf{LinDLLR}$-theory. As before, $\Omega^{n(*)}$ denotes $\Omega^{* \ldots *}$ with $n$ occurrences of $*$. Now we will prove the following auxiliary claim:

\begin{lemma}\label{c: lin}
Let $\Omega$ be a prime $\mathsf{LinDLLR}$-theory and $n$ a natural number. Then $\Omega \subseteq \Omega^{n(*)}$ or $\Omega^{n(*)} \subseteq \Omega$. 
\end{lemma}

For the sake of contradiction, assume that $\Omega \nsubseteq \Omega^{n(*)}$ and $\Omega^{n(*)} \nsubseteq \Omega$. So, there is $\alpha \in \Omega$\ such that $\alpha \notin \Omega^{n(*)}$, and there is $\beta \in \Omega^{n(*)}$ such that $\beta \notin \Omega$. We distinguish two possible cases. First, assume that $n=2k$. Then $\NOT^{n} \alpha \notin \Omega$ and $\NOT^{n} \beta \in \Omega$. Hence, we have $\alpha \AND \NOT^{n} \beta \in \Omega$ and $\NOT^{n} \alpha \OR \beta \notin \Omega$. This is a contradiction with the assumption that (L1) is in $\mathsf{LinDLLR}$. Second, assume that $n=2k+1$. Then $\NOT^{n} \alpha \in \Omega$ and $\NOT^{n} \beta \notin \Omega$. Then $\alpha \AND \NOT^{n} \alpha \in \Omega$ and $\beta \OR \NOT^{n} \beta \notin \Omega$ which is in contradiction with the assumption that (L2) is in $\mathsf{LinDLLR}$. This finishes the proof of Lemma \ref{c: lin}.

Now assume that $\alpha \vdash \beta \notin \mathsf{LinDLLR}$. By the standard procedure we can construct a prime $\mathsf{LinDLLR}$-theory $\Delta$ such that $\alpha \in \Delta$ and $\beta \notin \Delta$. Relative to  $\alpha \vdash \beta$ we construct  a Routley IM  $\mathcal{M}=\left\langle S, \circ, *,  i, e, V \right\rangle$ where $i$ is the set of all $L$-formulas; $e$ is the empty set; $S=\{i, e\} \cup \{\Delta^{n(*)} \mid n \in \mathbb{N} \}$; the star operation is extended to $S$\ by fixing $i^*=e$ and $e^*=i$; $\circ$ is defined as intersection; and $V(p)=\{\Omega \in S \mid p \in \Omega \}$. 

Note that Lemma \ref{c: lin} guarantees that $S$ is closed under intersection.  Since it is also closed under $*$ and satisfies all the desired properties we obtain that $\mathcal{M}$ is ineed a Routley IM. Lemma \ref{c: lin} further guarantees that $\mathcal{M}$ is linearly ordered. Now it can be shown by a straightforward induction that support corresponds to membership in this particular model.
\begin{lemma}
For any $\Omega \in S$\ and any $L$-formula $\alpha$, $\Omega \Vdash \alpha$\ in $\mathcal{M}$ iff $\alpha \in \Omega$.
\end{lemma}
It follows from this lemma that $\Delta \Vdash \alpha$ but $\Delta \nVdash \beta$, so $\alpha \vdash \beta$ is not valid in all linearly ordered IFs. This finishes the proof of Theorem \ref{p:compforlinrifs}.
\end{proof}
Now we will focus on those linearly ordered Routley IFs that validate double negation laws. As is shown in Proposition \ref{p:chardndemorg} these frames are characterized by general validity of $s=s^{**}$. If this condition is satisfied we say that the frame is involutive. Fig. \ref{f: involutivelinrims} shows schematically the structure of involutive linearly ordered Routley IFs.

\begin{figure}
\begin{center}
\begin{tabular}{c}
\begin{tikzpicture}
\node (a) at (0,0.1) {$\vdots$};
\node (b) at (0,0.5) {$\bullet$};
\node (c) at (0,-0.5) {$\bullet$};
\node (d) at (0,1.1) {$\vdots$};
\node (e) at (0,-0.9) {$\vdots$};
\node (f) at (0,1.5) {$\bullet$};
\node (g) at (0,-1.5) {$\bullet$};
\node (h) at (0,2.1) {$\vdots$};
\node (i) at (0,-1.9) {$\vdots$};
\node (j) at (0,2.5) {$i$};
\node (k) at (0,-2.5) {$e$};

\draw[->] (b) edge [bend left=80] (c);
\draw[->] (c) edge [bend left=80] (b);
\draw[->] (f) edge [bend left=80] (g);
\draw[->] (g) edge [bend left=80] (f);
\draw[->] (j) edge [bend left=80] (k);
\draw[->] (k) edge [bend left=80] (j);
\end{tikzpicture}
\end{tabular}
\end{center}
\caption{Involutive linearly ordered Roultey IFs}\label{f: involutivelinrims}
\end{figure}
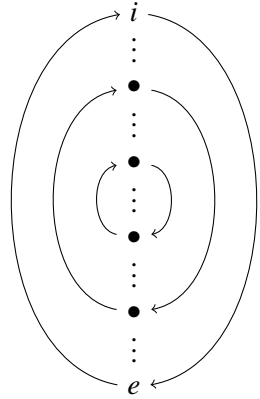

We will prove that the logic of all involutive linearly ordered Routley IFs coincides with what Dunn in \cite{Dunn99} called \textit{Kalman logic} ($\mathsf{KL}$). This logic is axiomatized by $\mathsf{DLLR}$ enriched with double negation laws (DN1) and (DN2), and the following axiom:
\begin{center}
\begin{tabular}{c}
(KA) $\alpha \AND \NOT \alpha \vdash  \beta \OR \NOT \beta$  \\
\end{tabular}
\end{center}





Interestingly, this logic is the $\{\AND, \OR, \NOT\}$-fragment of the relatively well-known and studied logic $\mathsf{R}$-mingle (see \cite{Dunn99}). We now prove completeness of this logic with respect to all involutive linearly ordered Routley IFs, and also with respect to the particular Routley IF $\mathcal{I}(\mathsf{KL})$ depicted in Fig. \ref{f: ikl}.
\begin{theorem}
$\alpha \vdash \beta$ is valid in all involutive linearly ordered Routley IFs iff $\alpha \vdash \beta \in \mathsf{KL}$. Moreover,  $\mathsf{KL}$ is also the set of all $L$-pairs that are valid in the Routley IF $\mathcal{I}(\mathsf{KL})$.
\end{theorem}
\begin{proof}
Notice that the axioms (L1) and (L2) collapse into (KA) in the context of the double negation laws. It follows that (KA) is valid in all involutive linearly ordered Routley IFs which gives us soundness, i.e. the right-to-left implication. In order to prove completeness, i.e. the left-to-right implication, one can show that $\mathsf{KL}$ is sound and complete with respect to the particular involutive linearly ordered Routley IF $\mathcal{I}(\mathsf{KL})$.

\begin{figure}
\begin{center}
\begin{tabular}{c}
\begin{tikzpicture}

\node (b) at (0,0.5) {$s$};
\node (c) at (0,-0.5) {$t$};
\node (j) at (0,1.5) {$i$};
\node (k) at (0,-1.5) {$e$};

\draw[->] (b) edge [bend left=80] (c);
\draw[->] (c) edge [bend left=80] (b);
\draw[->] (j) edge [bend left=80] (k);
\draw[->] (k) edge [bend left=80] (j);
\end{tikzpicture}
\end{tabular}
\end{center}
\caption{The Routley IF $\mathcal{I}(\mathsf{KL})$}\label{f: ikl}
\end{figure}
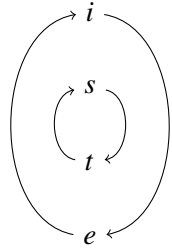
To prove this claim one can proceed as in Theorem \ref{p:compforlinrifs} and construct for a given $\alpha \vdash \beta \notin \mathsf{KL}$ a prime $\mathsf{KL}$-theory $\Delta$ such that $\alpha \in \Delta$ and $\beta \notin \Delta$. But now, due to the double negation laws, we obtain that $\Delta^{**}=\Delta$. Since $i$ is the set of all $L$-formulas and $e$ is the empty set, it follows that $\Delta \neq i$ and $\Delta \neq e$. Moreover, by the double negation laws we obtain $\NOT \alpha \notin \Delta^*$ and $\NOT \beta \in \Delta^*$, so we have also $\Delta^{*} \neq i$ and $\Delta^{*} \neq e$. If $\Delta \neq \Delta^*$ then the Routley IM described in the proof of Theorem \ref{p:compforlinrifs} has now the structure of $\mathcal{I}(\mathsf{KL})$. It can happen that $\Delta = \Delta^*$ but in that special case we can just take two copies of $\Delta$ and define the Routley star operation accordingly to obtain again the structure of $\mathcal{I}(\mathsf{KL})$.
\end{proof}

Note that since $\mathsf{KL}$ is the logic of one finite frame, it also follows from our proof that this logic is decidable. This is however not a new result because $\mathsf{KL}$ has been characterized, e.g. in \cite{Dunn99}, as a three-valued logic.

\section{The logic of all Routley IFs}\label{sec:logicofallRIF}

The logic of all Routley IFs can be axiomatized in the following way. It contains the axioms and rules of the system $\mathsf{DLLR}$ with a modified distributivity axiom and one of the DeMorgan laws:

\begin{center}
\begin{tabular}{ccccc}
(ID) $\alpha \vdash \alpha$ & (e$\AND_1$) $\alpha \AND \beta \vdash \alpha$ & (e$\AND_2$) $\alpha \AND \beta \vdash \beta$ & (i$\OR_1$) $\alpha \vdash \alpha \OR \beta$ & (i$\OR_2$) $\beta \vdash \alpha \OR \beta$ 
\end{tabular}
\end{center}

\begin{center}
\begin{tabular}{ccc}
(DM1) $\NOT (\alpha \AND \beta) \vdash \NOT \alpha \OR \NOT \beta$ & (DM2$^*$): $\NOT \NOT \alpha \AND \NOT \NOT \beta \vdash \NOT (\NOT \alpha \OR \NOT \beta)$ \\
\end{tabular}
\end{center}

\begin{center}
\begin{tabular}{c}
(D$^*$) $\alpha \AND (\NOT \beta \OR \NOT \gamma) \vdash (\alpha \AND \NOT \beta) \OR (\alpha \AND \NOT \gamma)$  \\
\end{tabular}
\end{center}

\begin{center}
\begin{tabular}{cccc}
(T) $\alpha \vdash \beta, \beta \vdash \gamma / \alpha \vdash \gamma$ & (i$\AND$) $\alpha \vdash \beta, \alpha \vdash \gamma / \alpha \vdash \beta \AND \gamma$ & (e$\OR$) $\alpha \vdash \gamma, \beta \vdash \gamma / \alpha \OR \beta \vdash \gamma$\\
\end{tabular}
\end{center}

\begin{center}
\begin{tabular}{c}
 (N) $\alpha \vdash \beta / \NOT \beta \vdash \NOT \alpha$  
\end{tabular}
\end{center}

Moreover, for each \textit{positive} number $k$ the system contains the following two axioms, which modify the characteristic axioms of $\mathsf{LinDLLR}$:
\begin{center}
\begin{tabular}{cc}
(L1$^*$) $\NOT \alpha \AND \NOT^{2k+1} \beta \vdash  \NOT^{2k+1}  \alpha \OR \NOT \beta$ & (L2$^*$) $\NOT \alpha \AND \NOT^{2k} \alpha \vdash \NOT \beta \OR \NOT^{2k} \beta$  \\
\end{tabular}
\end{center}

Let us call this logic $\mathsf{LRIF}$ (the logic of Routley information frames). Note that even though Routley star is often used in the context of relevant logics, the logic $\mathsf{LRIF}$ (similarly to the logics  $ \mathsf{LinDLLR}$ and $\mathsf{KL}$ discussed in the previous sections) is not ``relevant'' in a strict sense of the word because the axiom (L2$^*$) violates one of the essential features of relevant logics, namely the variable sharing property.

The rest of this paper is devoted to proving completeness of $\mathsf{LRIF}$ with respect to Routley IFs. The proof is similar to the proof of  Theorem \ref{p:compforlinrifs} with some additional complications. First, we have to replace the notion of prime theory with the following one.

\begin{definition}
 A set of $L$-formulas $\Omega$ is a negative prime $\mathsf{LRIF}$-theory of the just described system if it satisfies the following three conditions: (a) if $\alpha \in \Omega$ and $\alpha \vdash \beta \in \mathsf{LRIF}$ then $\beta \in \Omega$; (b)  if $\alpha \in \Omega$ and $\beta \in \Omega$ then $\alpha \AND \beta \in \Omega$; (c) if $\NOT \alpha \OR \NOT \beta \in \Omega$ then $\NOT \alpha \in \Omega$ or $\NOT \beta \in \Omega$. 
\end{definition}
In order to prove our main result, we will need several lemmas.
\begin{lemma}\label{l:negprimeext}
If $\alpha \vdash \beta \notin \mathsf{LRIF}$, then there is a negative prime $\mathsf{LRIF}$-theory $\Omega$ such that $\alpha \in \Omega$ and $\beta \notin \Omega$.
\end{lemma}
\begin{proof}
This result is obtained by standard Lindenbaum-style construction using the modified distributivity axiom (D$^*$).
\end{proof}

\begin{lemma}\label{l:nptclunderst}
If $\Omega$\ is a negative prime $\mathsf{LRIF}$-theory, then $\Omega^*$ is also a negative prime $\mathsf{LRIF}$-theory.
\end{lemma}
\begin{proof}
Assume that $\NOT \alpha \OR \NOT \beta \in \Omega^*$. So, $\NOT (\NOT \alpha \OR \NOT \beta) \notin \Omega$. Using the axiom (DM2$^*$) we obtain $\NOT \NOT \alpha \AND \NOT \NOT \beta \notin \Omega$. So, $\NOT \NOT \alpha \notin \Omega$ or $\NOT \NOT \beta \notin \Omega$ and thus $\NOT \alpha \in \Omega^*$ or $\NOT \beta \in \Omega^*$.
\end{proof}

\begin{lemma}\label{l: linorderneg}
If $\Omega$\ is a negative prime $\mathsf{LRIF}$-theory then $\Omega^{m(*)} \subseteq \Omega^{n(*)}$ or $\Omega^{n(*)} \subseteq \Omega^{m(*)}$, for all positive natural numbers $m, n$.
\end{lemma}
\begin{proof}
We can reason as in the proof of Lemma \ref{c: lin} using the axioms (L1$^*$) and (L2$^*$).
\end{proof}

\begin{lemma}\label{l: nptclint}
If $\Omega$ is a negative prime $\mathsf{LRIF}$-theory then $\Omega \cap \Omega^{n(*)}$ is a negative prime $\mathsf{LRIF}$-theory, for every positive natural number $n$.
\end{lemma}
\begin{proof}
Assume that we have $\NOT \alpha \notin \Omega \cap \Omega^{n(*)}$, and $\NOT \beta \notin \Omega \cap \Omega^{n(*)}$. Without loss of generality, assume that $\NOT \alpha \notin \Omega$ and $\NOT \beta \notin \Omega^{n(*)}$. Hence, $\alpha \in \Omega^{*}$ and $\beta \in \Omega^{n+1(*)}$. By Lemma \ref{l: linorderneg}, we obtain $\alpha, \beta \in \Omega^{*}$ or $\alpha, \beta \in \Omega^{n+1(*)}$, and thus $\NOT \alpha, \NOT \beta  \notin \Omega$ or $\NOT \alpha, \NOT \beta  \notin \Omega^{n(*)}$. Since, by Lemma \ref{l:nptclunderst}, $\Omega$ and $\Omega^{n(*)}$ are both negative prime $\mathsf{LRIF}$-theories, we obtain $\NOT \alpha \OR \NOT \beta  \notin \Omega$ or $\NOT \alpha \OR \NOT \beta  \notin \Omega^{n(*)}$. It follows that $\NOT \alpha \OR \NOT \beta \notin \Omega \cap \Omega^{n(*)}$.
\end{proof}

\begin{lemma}\label{l:starofintstars}
If $\Omega$\ is a negative prime $\mathsf{LRIF}$-theory then $(\Omega \cap \Omega^{n(*)})^*=\Omega^{*}$ or $(\Omega \cap \Omega^{n(*)})^*=\Omega^{n +1(*)}$, for every positive natural number $n$.
\end{lemma}
\begin{proof}
Since $\Omega \cap \Omega^{n(*)} \subseteq \Omega$, we obtain $\Omega^{*}\subseteq (\Omega \cap \Omega^{n(*)})^*$. By the same reasoning, $\Omega^{n+1 (*)}\subseteq (\Omega \cap \Omega^{n(*)})^*$. We have to prove that $(\Omega \cap  \Omega^{n(*)})^* \subseteq \Omega^{*}$ or $(\Omega \cap  \Omega^{n(*)})^* \subseteq \Omega^{n +1(*)}$. For the sake of contradiction assume that there are $\alpha, \beta$\ such that $\alpha, \beta  \in (\Omega \cap \Omega^{n(*)})^*$, $\alpha \notin \Omega^{n+1(*)}$ and $\beta \notin \Omega^{*}$. On one hand, we obtain $\NOT \alpha, \NOT \beta  \notin \Omega \cap \Omega^{n(*)}$. By Lemma \ref{l: nptclint}, $\Omega \cap \Omega^{n(*)}$ is a negative prime $\mathsf{LRIF}$-theory and hence $\NOT \alpha \OR \NOT \beta  \notin \Omega \cap \Omega^{n(*)}$. On the other hand, we obtain $\NOT \alpha \in \Omega^{n(*)}$ and $\NOT \beta \in \Omega$. It follows that $\NOT \alpha \OR \NOT \beta \in \Omega \cap \Omega^{n(*)}$, which is a contradiction.
\end{proof}

\begin{theorem}
$\alpha \vdash \beta$ is valid in all Routley IFs iff $\alpha \vdash \beta \in \mathsf{LRIF}$.
\end{theorem}
\begin{proof}
We will prove completeness, i.e. the left-to-right implication. Assume that $\alpha \vdash \beta \notin \mathsf{LRIF}$. By Lemma \ref{l:negprimeext} there is a negative prime $\mathsf{LRIF}$-theory $\Delta$ such that $\alpha \in \Delta$ and $\beta \notin \Delta$. Relative to  $\alpha \vdash \beta$ we construct  a Routley IM  $\mathcal{M}=\left\langle S, \circ, *,  i, e, V \right\rangle$ where $i$ is the set of all $L$-formulas; $e$ is the empty set; $S=\{i, e\} \cup \{\Delta^{m(*)} \cap \Delta^{n(*)} \mid m, n \in \mathbb{N} \}$; the star operation is defined as usual; $\circ$ is defined as intersection; and $V(p)=\{\Omega \in S \mid p \in \Omega \}$. The structure of the model is illustrated with Fig. \ref{f: can model} (for the case where $\Delta^* \subseteq \Delta^{**}$).

Note that  $S$ is closed under intersection and $*$. In particular, it is closed under intersection because of Lemmas \ref{l: linorderneg} and \ref{l: nptclint} and it is closed under $*$ because of Lemmas \ref{l:nptclunderst} and \ref{l:starofintstars}. Moreover, Lemma  \ref{l:starofintstars} guarantees that $\mathcal{M}$ is a Routley IM. Again, it can be shown by straightforward induction that $\Omega \Vdash \alpha$\ in $\mathcal{M}$ iff $\alpha \in \Omega$. It follows that $\Delta \Vdash \alpha$ but $\Delta \nVdash \beta$, so $\alpha \vdash \beta$ is not valid in all IFs.
\end{proof}

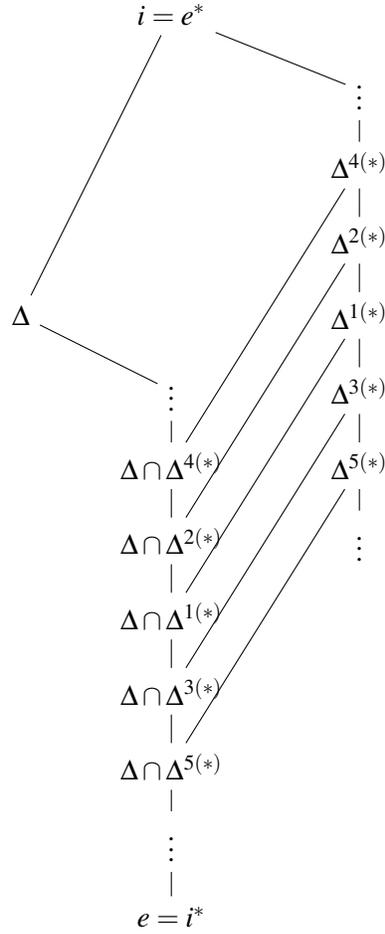
\begin{figure}
\begin{center}
\begin{tabular}{c}
\begin{tikzpicture}

\node (a) at (-2,0) {$\Delta$};
\node (b) at (2.5,0) {$\Delta^{1(*)}$};
\node (c) at (2.5,1) {$\Delta^{2(*)}$};
\node (d) at (2.5,-1) {$\Delta^{3(*)}$};
\node (e) at (2.5,2) {$\Delta^{4(*)}$};
\node (f) at (2.5,-2) {$\Delta^{5(*)}$};
\node (g) at (0,-4) {$\Delta \cap \Delta^{1(*)}$};
\node (h) at (0,-3) {$\Delta \cap \Delta^{2(*)}$};
\node (i) at (0,-5) {$\Delta \cap \Delta^{3(*)}$};
\node (j) at (0,-2) {$\Delta \cap \Delta^{4(*)}$};
\node (k) at (0,-6) {$\Delta \cap \Delta^{5(*)}$};
\node (l) at (2.5,3) {$\vdots$};
\node (m) at (2.5,-3) {$\vdots$};
\node (n) at (0,-1) {$\vdots$};
\node (o) at (0,-7) {$\vdots$};
\node (p) at (0,-8) {$e=i^*$};
\node (q) at (0,4) {$i=e^*$};

\draw[-] (b) edge (c);
\draw[-] (c) edge (e);
\draw[-] (b) edge (d);
\draw[-] (d) edge (f);
\draw[-] (e) edge (l);
\draw[-] (f) edge (m);

\draw[-] (a) edge (n);
\draw[-] (n) edge (j);
\draw[-] (j) edge (h);
\draw[-] (h) edge (g);
\draw[-] (g) edge (i);
\draw[-] (i) edge (k);
\draw[-] (k) edge (o);

\draw[-] (b) edge (g);
\draw[-] (c) edge (h);
\draw[-] (d) edge (i);
\draw[-] (e) edge (j);
\draw[-] (f) edge (k);

\draw[-] (p) edge (o);

\draw[-] (q) edge (a);
\draw[-] (q) edge (l);

\end{tikzpicture}
\end{tabular}
\end{center}
\caption{The structure of a Routley IM that provides a counterexample to an $\mathsf{LRIF}$-unprovable $\alpha \vdash \beta$}\label{f: can model}
\end{figure}

\section{Conclusion}

The goal of this paper was to explore the behaviour of Routley star in the context of ``information-based semantics'' in which states are not in general prime. We succeeded in characterizing syntactically the logic of all information frames equipped with Routley star. We also showed that if double negation laws are added, the resulting logic coincides  with Kalman logic that is known to be the $\{\AND, \OR, \NOT\}$-fragment of $\mathsf{R}$-mingle. We observed that it can be characterized as the logic of involutive  linearly ordered information frames. We also characterized the logic of all linear frames without the requirement of involution.

\bibliographystyle{eptcs}
\bibliography{biblio}

\end{document}